\documentclass[12pt]{article}
\usepackage{amsthm,amsfonts,amsmath,amssymb,mathrsfs,graphicx,url,tikz,hyperref}

\title{A Relativistic GRW Flash Process With Interaction}

\author{
Roderich Tumulka\footnote{Fachbereich Mathematik, Eberhard-Karls-Universit\"at, Auf der Morgenstelle 10, 72076 T\"ubingen, Germany. E-mail: roderich.tumulka@uni-tuebingen.de}
}
\date{February 11, 2020}

\newcommand{\be}{\begin{equation}}
\newcommand{\ee}{\end{equation}}
\newcommand{\Hilbert}{\mathscr{H}}

\newcommand{\scp}[2]{\langle #1|#2 \rangle}
\newcommand{\Bscp}[2]{\Bigl\langle #1\Big|#2 \Bigr\rangle}

\newcommand{\future}{\mathrm{future}}
\newcommand{\past}{\mathrm{past}}
\newcommand{\sdist}{\text{s-dist}}
\newcommand{\In}{\mathrm{In}}
\newcommand{\Out}{\mathrm{Out}}

\newcommand{\CCC}{\mathbb{C}}

\newcommand{\HHH}{\mathbb{H}}
\newcommand{\MMM}{\mathbb{M}}
\newcommand{\NNN}{\mathbb{N}}
\newcommand{\PPP}{\mathbb{P}}
\newcommand{\RRR}{\mathbb{R}}

\newcommand{\uell}{\underline{\ell}}
\newcommand{\uk}{\underline{k}}
\newcommand{\un}{\underline{n}}
\newcommand{\ux}{\underline{x}}
\newcommand{\uX}{\underline{X}}

\newcommand{\sA}{\mathscr{A}}
\newcommand{\sB}{\mathscr{B}}
\newcommand{\sC}{\mathscr{C}}

\newcommand{\sI}{\mathscr{I}}

\newcommand{\sN}{\mathscr{N}}

\newcommand{\sS}{\mathscr{S}}

\theoremstyle{theorem}
\newtheorem{prop}{Proposition}
\theoremstyle{definition}

\begin{document}

\maketitle

\begin{abstract}
In 2004, I described a relativistic version of the Ghirardi-Rimini-Weber (GRW) model of spontaneous wave function collapse for $N$ non-interacting distinguishable particles. Here I present a generalized version for $N$ interacting distinguishable particles. Presently, I do not know how to set up a similar model for indistinguishable particles or a variable number of particles. 
The present interacting model is constructed from a given interacting unitary Tomonaga-Schwinger type evolution between spacelike hypersurfaces, into which discrete collapses are inserted. I assume that this unitary evolution is interaction-local (i.e., no interaction at spacelike separation). The model is formulated in terms of Bell's flash ontology but is also compatible with Ghirardi's matter density ontology. It is non-local and satisfies microscopic parameter independence and no-signaling; it also works in curved space-time; in the non-relativistic limit, it reduces to the known non-relativistic GRW model.

\medskip

\end{abstract}

\bigskip

\begin{center}
{\it To the memory of GianCarlo Ghirardi (1935--2018)}
\end{center}

\section{Introduction}

In this paper, I describe a relativistic model of spontaneous wave function collapse for $N$ distinguishable particles with interaction, thereby generalizing my 2004 model without interaction \cite{Tum06,Tum06c,Tum09}. The model involves, like the GRW model \cite{GRW86,Bell87} on which it is based, discrete jumps of the wave function with unitary evolution in between. I will describe the model in terms of Bell's flash ontology \cite{Bell87,bmgrw,Tum18} by specifying the joint probability distribution of all flashes, but it could also be set up using Ghirardi's matter density ontology \cite{BGG95,bmgrw} along the lines described in \cite{BDGGTZ14}. Neither choice of ontology makes the problem easier or more difficult. In the model, one can, as already suggested by Aharonov and Albert in 1981 \cite{AA81}, associate a wave function $\psi_\Sigma$ with every spacelike hypersurface $\Sigma$. For every flash in the past of $\Sigma$, a collapse operator gets applied to the wave function; thus, each collapse affects the wave function everywhere in the universe, although the model is fully relativistic. Since the flashes occur randomly, $\psi_\Sigma$ is a random wave function and can thus be regarded as subject to a stochastic time evolution. In contrast to Bohmian mechanics but like the non-interacting 2004 model, the present model does not invoke a preferred foliation (i.e., slicing) of space-time into spacelike hypersurfaces. 

Since, unlike Bohmian mechanics, the model does not involve trajectories in space-time, the word ``particle'' should not be taken literally. Rather, in this paper it means a space-time variable in the wave function (or in the configuration PVM on Hilbert space).

The interaction is incorporated by assuming the unitary part of the time evolution as given and including interaction. More precisely, I assume that a unitary Tomonaga-Schwinger type evolution between spacelike hypersurfaces is given and describe how to insert collapses in between unitary evolution operators. I assume that the unitary evolution is relativistic and interaction-local (i.e., involves no interaction terms between spacelike separated regions, see below). For example, such an evolution is rigorously known for $N$ Dirac particles in 1+1 dimensions with zero-range interaction \cite{Lie15a,Lie15c,LN15}. For another example, the $N$ particles could be taken to interact through a quantized field (which will neither be associated with local beables by itself nor with collapses).\footnote{However, quantized fields are usually mathematically ill defined due to ultraviolet divergence. It would be of interest to study carefully whether one of the few mathematically well defined evolutions for quantum fields (such as \cite{GJ68} in 1+1 dimensions) can be put to work here.} If the unitary evolution is non-interacting, then the model reduces, up to small deviations, to the 2004 version. In particular, the model is non-local, i.e., two spacelike separated events $a$ and $b$ can influence each other, although there is no fact about the direction of the influence (whether $a$ influenced $b$ or $b$ influenced $a$) \cite{Tum07}. Also like the 2004 version, the model obeys, up to small deviations, the condition that the distribution of the flashes up to a given spacelike hypersurface $\Sigma$ does not depend on external fields in the future of $\Sigma$; this condition is a microscopic analog of the condition known as ``parameter independence.''

Much of the difficulty of devising a model with interaction arises from the fact that the collapse operators associated with different particles do not generally commute, whereas they do in the non-interacting case. This leads to a question of how to order the operators in the formula defining the joint probability density of the flashes, all the while with a need to ensure that the density integrates up to 1. The procedure proposed here is, roughly speaking, based on ordering the operator factors associated with two flashes in the temporal order when they are timelike separated, while the operators essentially still commute when the flashes are spacelike separated, except for certain details arising from the width $\sigma$ of the collapses. 

Another difficulty arises precisely from the use of smeared-out collapse operators of width $\sigma$. Partly due to the use of different operator orderings depending on the space-time locations of the flashes, it turned out relevant to cut off the tails of the profile function, usually a Gaussian function of width $\sigma$, to ensure it vanishes exactly outside a certain admissible region. In fact, ``cut off the tails'' is a shorthand for a somewhat more involved procedure that will also change the shape of the profile function (away from a Gaussian shape) in the region where it does not vanish, as I will explain in Section~\ref{sec:simple}. To make the model work, these several difficulties must jointly be dealt with.

Like the original GRW model, the present model has two parameters, the width $\sigma$ of the collapse and the collapse rate $\lambda$ per particle (or, equivalently, the expected waiting time $\tau=1/\lambda$ for a collapse for a given particle). For our purposes, the waiting time is the relativistic timelike distance between two flashes associated with the same particle. We assume here the values suggested by GRW \cite{GRW86}, $\sigma \approx 10^{-7}$~m and $\tau \approx 10^{16}$ s. The empirical predictions of the model are presumably, like those of the original GRW model, too close to those of standard quantum mechanics to allow for an experimental test with present technology; a careful study of its empirical predictions, its deviations from the original GRW model, and possible experimental tests would be of interest.

While the considerations of this paper also work in curved space-time, they will be formulated for Minkowski space-time $\MMM$.

Let me mention other proposals of relativistic collapse theories: Early attempts at a relativistic version of continuous spontaneous localization (CSL) \cite{Pe90,Pe99} are divergent and lead to infinite energy increase; see also \cite{BP19}. A regularized relativistic version was developed by Bedingham and Pearle \cite{Bed10,Bed11,Pe15}. 
A model due to Dowker and Henson \cite{Fay02} lives on a discrete space-time and is relativistic in the appropriate lattice sense. A relativistic model due to Tilloy \cite{Til} is based on starting from a standard quantum field theory in a suitable regime, tracing out certain degrees of freedom, obtaining a master equation for the remaining ones, and finally using an unraveling of that master equation that should be empirically equivalent to the original quantum field theory in the regime considered.

The remainder of this paper is organized as follows. In Section~\ref{sec:background}, I review the non-interacting model. In Section~\ref{sec:assumptions}, I describe the assumptions made on the unitary part of the time evolution. In Section~\ref{sec:model}, I define the interacting model. In Section~\ref{sec:properties}, I discuss some of its properties. 

\section{Review of the Non-Interacting Version}
\label{sec:background}

We begin with a brief summary of the 2004 model for $N$ distinguishable non-interacting particles. 

I will specify the joint distribution of the first $n_i$ flashes for each particle number $i\in\{1,\ldots,N\}$. Let $X_{ik}$ with $i\in\{1,\ldots,N\}$, $k\in \{1,\ldots,n_i\}$ be the random space-time points at which the flashes occur. Let $\uX$ denote the collection of all $X_{ik}$ with $1\leq i\leq N$ and $1\leq k\leq n_i$, likewise $\ux$ the collection of the space-time points $x_{ik}$, and
\be\label{dux}
d\ux= \prod_{i=1}^N \prod_{k=1}^{n_i} d^4x_{ik}
\ee
the volume element in $4\nu$ dimensions (meaning either an infinitesimal set or its volume) with $\nu:=n_1+\ldots+n_N$.
For each $i$, let $x_{i0}$ be a given ``seed'' flash. The distribution of $\uX$ is of the form
\be\label{PPPdef1}
\PPP\bigl(\uX \in d\ux \bigr) = \scp{\psi_0}{D(\ux)|\psi_0} \, d\ux
\ee
with operators $D$ to be specified below and $\psi_0$ a wave function on a surface $\Sigma_0$ playing the role of an initial surface. In particular, the distribution of $\uX$ is associated with a POVM $G(d\ux) = D(\ux) \, d\ux$ with density $D(\ux)$. Let $\Hilbert_{1\Sigma}$ be the 1-particle Hilbert space associated with the spacelike surface $\Sigma$, and $\Hilbert_{10}:=\Hilbert_{1\Sigma_0}$, so $\psi_0\in \Hilbert_0 :=\Hilbert_{10}^{\otimes N}$. For each $i\in\{1,\ldots,N\}$, let $U^{\Sigma'}_{i\Sigma}$ be the unitary time evolution of particle $i$ from the spacelike surface $\Sigma$ to the spacelike surface $\Sigma'$. For all particles together, the unitary time evolution is $U^{\Sigma'}_{\Sigma}= U^{\Sigma'}_{1\Sigma} \otimes \cdots \otimes U^{\Sigma'}_{N\Sigma}$. We write $U_0^{\Sigma'}$ for $U_{\Sigma_0}^{\Sigma'}$.

\begin{figure}[h]
\begin{center}
\begin{tikzpicture}
\draw (-2.5,-0.5) -- (2.5,-0.5);
\draw (-2,-1) -- (-2,2.8);
\node at (2,-0.75) {space};
\node at (-1.5,2.6) {time};
\draw [very thick, domain=-2.5:2.5, samples=50] plot (\x, {sqrt(1+\x*\x))} );
\node at (1.3,2.3) {$\HHH_y(s)$};
\draw (0,1) -- (-1.5,1);
\draw (0,0) -- (-1.5,0);
\draw[<->] (-1.4,0) -- (-1.4,1);
\node at (-1.55,0.5) {$s$}; 
\draw[dashed] (0,0) -- (2.5,2.5);
\draw[dashed] (0,0) -- (-2.5,2.5);
\filldraw (0,0) circle [radius=0.07];
\node at (0.2,-0.2) {$y$};
\end{tikzpicture}
\end{center}
\caption{In Minkowski space-time, the surface of constant timelike distance $s$ from $y$ in the future of $y$, $\HHH_y(s)$, has the shape of a hyperboloid that is asymptotic to the future light cone of $y$ (dashed).}
\label{fig:H}
\end{figure}
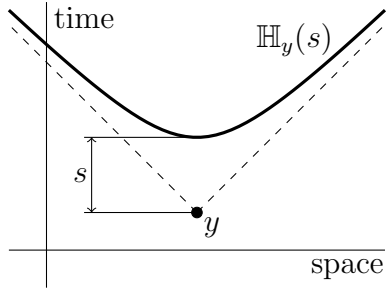

Let $\HHH_y(s)$ be the surface of constant timelike distance $s$ from $y\in\MMM$ in the future of $y$ (henceforth called a hyperboloid, see Figure~\ref{fig:H}); we also write $|\cdot|$ for the invariant (proper) length of a timelike 4-vector, $\HHH_y(x):=\HHH_y(|x-y|)$ for the hyperboloid containing $x\in \future(y)$, and\footnote{For definiteness, we take $\future(y)$ to be a closed set (i.e., the ``causal future''), including the future light cone and $y$ itself. Likewise for the past.}
\be\label{HHHikdef}
\HHH_{ik} := \HHH_{x_{ik-1}}(x_{ik})\,.
\ee
Let $\tilde g_{yx}$ be the Gaussian function centered at $x$ along the hyperboloid $\HHH_y(x)$,
\be\label{tildegdef}
\tilde{g}_{yx}(z) := \exp\Biggl(-\frac{\sdist_{\HHH_y(x)}(x,z)^2}{4\sigma^2} \Biggr)\,,
\ee
where $\sdist_\Sigma$ means the spacelike distance along $\Sigma$,
and $g_{yx}$ the version normalized in $x$,
\be\label{gdef}
g_{yx}(z) := \frac{1}{\|\tilde{g}_{yz}\|} \tilde{g}_{yx}(z)\,,
\ee
where, for a spacelike surface $\Sigma$ and $f:\Sigma\to \CCC$,
\be
\|f\|:= \Bigl( \int_{\Sigma} d^3x \, |f(x)|^2 \Bigr)^{1/2}
\ee
is the $L^2$ norm and $d^3x$ means the invariant volume of a 3-surface element (defined by the 3-metric on $\Sigma$).
For the multiplication operator by $g_{yx}$ on $\Sigma=\HHH_y(x)$, we write $P(g_{yx})$. To the flash $x_{ik}$ we associate the collapse operator
\be\label{Kdef1}
K(x_{ik}):= U_{i\HHH_{ik}}^0 \, P(g_{x_{ik-1}x_{ik}})\, U_{i0}^{\HHH_{ik}}\,,
\ee
and to all flashes together the operator
\be\label{Ldef1}
L(\ux):= \bigotimes_{i=1}^N \prod_{k=1}^{n_i} K(x_{ik})
\ee
with the order in the product so that $k$ increases from right to left. Then set
\be\label{Ddef1}
D(\ux) := \biggl( \frac{1}{\tau^\nu} \prod_{i=1}^N \prod_{k=1}^{n_i} 1_{x_{ik}\in \future(x_{ik-1})} e^{-|x_{ik}-x_{ik-1}|/\tau} \biggr)  L(\ux)^\dagger \, L(\ux)\,.
\ee
It was shown in \cite{Tum06} (and it follows from the proofs below that apply to the more general interacting case) that
\be
\int_{\MMM^\nu} \hspace{-3mm} d\ux \, D(\ux) = I
\ee
with $I$ the unit operator; as a consequence, $\PPP$ is a probability distribution for every $\psi_0\in \Hilbert_0$ with $\|\psi_0\|=1$.

Equivalently, $G(d\ux) = \otimes_{i=1}^N G_i(d^4x_{i1}\times \cdots \times d^4x_{in_i})$ and $D(\ux) = \otimes_{i=1}^N D_i(x_{i1},\ldots,x_{in_i})$ with $G_i(d^4x_{i1}\times \cdots\times d^4x_{in_i})= D_i(x_{i1},\ldots,x_{in_i}) \, d^4x_{i1}\cdots d^4x_{in_i}$ and 
\begin{multline}
D_i(x_{i1},\ldots,x_{in_i})= \biggl(\frac{1}{\tau^{n_i}}\prod_{k=1}^{n_i} 1_{x_{ik}\in \future(x_{ik-1})} e^{-|x_{ik}-x_{ik-1}|/\tau}\biggr)\,\times\\
K(x_{i1})^\dagger \cdots K(x_{in_i})^\dagger K(x_{in_i}) \cdots K(x_{i1})\,.
\end{multline}
If $\psi_0$ factorizes into a tensor product $\otimes_{i=1}^N \psi_{0i}$, then $\PPP$ will factorize, and the flashes for different $i$ will be independent of each other. But in general, even though $G$ factorizes, $\PPP$ will not factorize, which leads to non-local correlations between the flashes associated with different $i$.

We now prepare for defining the novel interacting version of the model.

\section{Assumptions}
\label{sec:assumptions}

In this article, I put no emphasis on mathematical rigor. But the reasoning is actually rigorous if the assumption is satisfied that the unitary evolution is defined not only between Cauchy surfaces but also to hyperboloids or surfaces consisting of pieces of hyperboloids. For example, a sufficient class of surfaces would be the set $\sS$ of those sets that are intersected exactly once by every timelike straight line; in the following, I will simply say ``spacelike surface'' for any $\Sigma\in\sS$. For massive free Dirac particles, it is known \cite{DP03} (see also \cite{Tum09}) that the unitary evolution is also defined from a Cauchy surface to a hyperboloid, so it seems plausible that it is also defined between any two surfaces belonging to $\sS$.

So, we assume that the unitary part of the time evolution is given by a Tomonaga--Schwinger type evolution, more precisely, by a \emph{unitary hypersurface evolution} \cite{LT19} between spacelike surfaces. That is, we assume that with every spacelike surface $\Sigma\in\sS$ there is associated a Hilbert space $\Hilbert_\Sigma$ (see \cite{LT19} for examples), and that for any two spacelike surfaces $\Sigma,\Sigma'$ we are given a unitary isomorphism $U^{\Sigma'}_\Sigma:\Hilbert_\Sigma \to \Hilbert_{\Sigma'}$ representing the time evolution without collapses, such that 
\be
U^\Sigma_\Sigma = I, ~~~U^{\Sigma''}_{\Sigma'} \, U^{\Sigma'}_{\Sigma} = U^{\Sigma''}_{\Sigma}
\ee
for all $\Sigma,\Sigma'$, and $\Sigma''$. Moreover, for each $\Sigma$ we are given a position PVM $P_\Sigma$ (``configuration observable'') on $\Sigma^N$ acting on $\Hilbert_\Sigma$. This completes the definition of ``unitary hypersurface evolution.''
For a function $f:\Sigma^N\to\RRR$, we define the associated multiplication operator
\be
P(f) := P_\Sigma(f) := \int_{\Sigma^N} P_\Sigma(d^3x_1\times \cdots \times d^3x_N) \, f(x_1,\ldots,x_N)\,. 
\ee
Of the unitary evolution we assume 

\bigskip

\noindent {\bf Interaction locality} (IL) \cite{LT19}: {\it For any two spacelike hypersurfaces $\Sigma,\Sigma'$, any set $A\subseteq \Sigma\cap \Sigma'$ in the overlap, and any $i\in\{1,\ldots,N\}$,}
\be\label{IL}
P_{\Sigma'}\Bigl((\Sigma')^{i-1}\times A \times (\Sigma')^{N-i-1}\Bigr)
= U^{\Sigma'}_{\Sigma} \, P_\Sigma\Bigl( \Sigma^{i-1}\times A \times \Sigma^{N-i-1} \Bigr) \, U^\Sigma_{\Sigma'}\,.
\ee

\bigskip

The condition expresses that the unitary evolution includes no interaction term between spacelike separated regions. Specifically, the unitary evolution from $\Sigma$ to $\Sigma'$ acts like the identity on $\Sigma \cap \Sigma'$. Here are some consequences of IL:
\begin{enumerate}
\item Fix a function $f$ on $\Sigma\cap \Sigma'$ and a label $i\in\{1,\ldots,N\}$, and let $P(f)$ be the associated multiplication operator in the $i$-th variable; more precisely, let
\begin{subequations}
\begin{align}
P_\Sigma(f) &:= \int\limits_{x\in \Sigma\cap \Sigma'} \hspace{-2mm} P_\Sigma\Bigl(\Sigma^{i-1} \times d^3x \times \Sigma^{N-i-1}\Bigr) \, f(x)\,,\\ 
P_{\Sigma'}(f) &:= \int\limits_{x\in \Sigma\cap \Sigma'} \hspace{-2mm} P_{\Sigma'}\Bigl(\Sigma^{\prime (i-1)} \times d^3x \times \Sigma^{\prime (N-i-1)}\Bigr) \, f(x)\,.
\end{align}
\end{subequations}
Then
\be\label{ILf}
P_{\Sigma'}(f) = U^{\Sigma'}_{\Sigma} \, P_\Sigma(f) \, U^\Sigma_{\Sigma'}\,.
\ee
That is because, setting $A=d^3x$ in \eqref{IL},
\begin{subequations}
\begin{align}
U^{\Sigma'}_{\Sigma} \, P_\Sigma(f) \, U^\Sigma_{\Sigma'}&=\int_{\Sigma\cap \Sigma'} U^{\Sigma'}_{\Sigma} \, P_\Sigma(\Sigma^{i-1} \times d^3x \times \Sigma^{N-i-1})\, U^\Sigma_{\Sigma'} \, f(x)\\
&=\int_{\Sigma\cap \Sigma'} P_{\Sigma'}(\Sigma^{\prime (i-1)} \times d^3x \times \Sigma^{\prime (N-i-1)}) \, f(x)\\[2mm]
&= P_{\Sigma'}(f)\,.
\end{align}
\end{subequations}
\item Let $A:= \Sigma\cap \Sigma'$, $B:=\Sigma\setminus A$, $B':= \Sigma' \setminus A$. Then 
\be\label{ILB}
P_{\Sigma'}\Bigl((\Sigma')^{i-1}\times B' \times (\Sigma')^{N-i-1}\Bigr)
= U^{\Sigma'}_{\Sigma} \, P_\Sigma\Bigl( \Sigma^{i-1}\times B \times \Sigma^{N-i-1} \Bigr) \, U^\Sigma_{\Sigma'}\,.
\ee
That is because, using the normalization $P_{\Sigma'}(\Sigma^{\prime N})=I$,
\begin{subequations}
\begin{align}
&P_{\Sigma'}\Bigl((\Sigma')^{i-1}\times B' \times (\Sigma')^{N-i-1}\Bigr)\\
&=I-P_{\Sigma'}\Bigl((\Sigma')^{i-1}\times A \times (\Sigma')^{N-i-1}\Bigr)\\
&\stackrel{\eqref{IL}}{=} I-U^{\Sigma'}_{\Sigma} \, P_\Sigma\Bigl( \Sigma^{i-1}\times A \times \Sigma^{N-i-1} \Bigr) \, U^\Sigma_{\Sigma'}\\
&= U^{\Sigma'}_{\Sigma} \, P_\Sigma\Bigl( \Sigma^{i-1}\times B \times \Sigma^{N-i-1} \Bigr) \, U^\Sigma_{\Sigma'}\,.
\end{align}
\end{subequations}
\end{enumerate}

\section{Interacting Model}
\label{sec:model}

\subsection{A Simple Case}
\label{sec:simple}

As a warm-up we consider the simple case of $N=2$ particles and limit our attention to the first flash for each particle. That is, we define the joint distribution of two flashes, $X_1$ and $X_2$, each associated with a different particle. We take as given a seed flash for each particle, $y_1$ and $y_2$, and an initial wave function $\psi_0$. We postulate that the joint distribution is of the form
\be\label{PPPdef2}
\PPP\bigl(X_1\in d^4x_1, X_2\in d^4x_2\bigr) = \scp{\psi_0}{D(x_1,x_2)|\psi_0} \, d^4x_1 \, d^4x_2
\ee
with positive operator-valued density
\be\label{Ddef2}
D(x_1,x_2) =1_{x_1 \in \future(y_1)} 1_{x_2\in \future(y_2)} \frac{1}{\tau^2} e^{-|x_1-y_1|/\tau}e^{-|x_2-y_2|/\tau} \, L(x_1,x_2)^\dagger \, L(x_1,x_2)\,,
\ee
where $L$ will be defined below. As before, the distribution of $(X_1,X_2)$ is determined by a POVM $G(dx_1\times dx_2) = D(x_1,x_2) \, dx_1 \, dx_2$ with density $D(x_1,x_2)$. We will consider two relevant hyperboloids,
\be
\HHH_i:= \HHH_{y_i}(x_i)~~~\text{with }i=1,2,
\ee
and a profile function $g_{yx}(z)$ on $\HHH_y(x)$ that has a bump shape around $x$; the first thought would be to use the Gaussian function $\tilde g_{yx}$ centered at $x$ given by \eqref{tildegdef}, but we will refine this choice later. We write $g_{yxi}$ for the function $g_{yx}$ applied to the $i$-th variable; that is, $g_{yx}$ is a function on a 3-surface $\Sigma$, and $g_{yxi}$ a function on $\Sigma^N$ (here with $N=2$). 

A basic difference between the interacting and the non-interacting case is that in the non-interacting case, one can evolve particle 1 to $\Sigma_1$ and independently particle 2 to another surface $\Sigma_2$; in the interacting case, we can only evolve all particles \emph{jointly} to a certain surface.
Let us consider two candidates for $L(x_1,x_2)$,
\begin{subequations}
\begin{align}
L^{(21)}(x_1,x_2) &:=  U^0_{\HHH_2} P_{\HHH_2}(g_{y_2x_22}) U^{\HHH_2}_{\HHH_1} P_{\HHH_1}(g_{y_1x_11}) \, U^{\HHH_1}_0\,,\label{L21}\\
L^{(12)}(x_1,x_2) &:= U^0_{\HHH_1} P_{\HHH_1}(g_{y_1x_11}) U^{\HHH_1}_{\HHH_2} P_{\HHH_2}(g_{y_2x_22}) \, U^{\HHH_2}_0\,.\label{L12}
\end{align}
\end{subequations}
We can think of each of \eqref{L21} and \eqref{L12} as a product of two multiplication operators, each Heisenberg-evolved to the initial surface $\Sigma_0$. Since the unitary evolution does not commute with multiplication operators, the two multiplication operators on different surfaces do not commute with each other, so \eqref{L21} and \eqref{L12} are in general not equal, with two relevant exceptions: First, in the absence of interaction, the unitary evolution factorizes, and the two multiplication operators commute because they act on different factors. Second, if the supports of $g_{y_1x_1}$ and $g_{y_2x_2}$ are spacelike separated (i.e., if every point in the one set is spacelike from every point in the other), then they commute by virtue of IL. (It may appear pointless to talk about the support of $g_{yx}$ if $g_{yx}$ is a Gaussian because then its support is the entire surface $\HHH_y(x)$; but we will later cut off the Gaussian tails to create smaller supports.)

So, for the purpose of defining the operator $L(x_1,x_2)$, we are confronted with a problem of operator ordering. Roughly speaking, we choose the ordering according to the temporal ordering of $x_1$ and $x_2$: For $x_1$ in the past of $x_2$, we choose $L=L^{(21)}$ and vice versa. But the exact definition is a little more complicated, partly because we need to consider the support of $g_{y_1x_1}$, not just the point $x_1$.

To this end, we subdivide each $\HHH_i$ into two parts (see Figure~\ref{fig:PFH}),
\be\label{PFHdef}
F_i:= \HHH_i \cap \future(\HHH_{3-i})\,,~~~~P_i:=\HHH_i\cap \past(\HHH_{3-i})\,.
\ee
(Note that the interface $\HHH_1\cap\HHH_2$ has measure 0 in $\HHH_1$ as well as in $\HHH_2$, except if $y_1=y_2$ and $\tau_1=\tau_2$, which happens with probability 0. Ignoring sets of measure 0, we can pretend that $F_i$ and $P_i$ form a partition of $\HHH_i$.)

\begin{figure}[h]
\begin{center}
\begin{tikzpicture}
\draw (-2.5,-0.5) -- (2.5,-0.5);
\draw (-2,-1) -- (-2,2.8);
\draw [very thick, domain=-2.5:2.5, samples=50] plot (\x, {-0.5+sqrt(1+(\x+0.8)*(\x+0.8)))} );
\node at (-0.7,0.15) {$\HHH_1$};
\draw [very thick, domain=-2.5:2.5, samples=50] plot (\x, {-0.3+sqrt(0.1+(\x-0.5)*(\x-0.5))} );
\node at (1,-0.2) {$\HHH_2$};
\node at (-1.65,0.45) {$P_1$};
\node at (1.2,2.2) {$F_1$};
\node at (2.1,0.8) {$P_2$};
\node at (-1.5,2.2) {$F_2$};
\end{tikzpicture}
\end{center}
\caption{Two hyperboloids $\HHH_1,\HHH_2$ are each subdivided according to \eqref{PFHdef} into two 3-regions $F_i$ and $P_i$, above and below the other.}
\label{fig:PFH}
\end{figure}
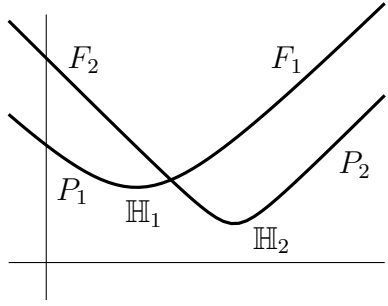

For any set $A\subseteq \HHH_y(x)$, set
\be
\|f\|_A = \Bigl(\int_A d^3z \, |f(z)|^2 \Bigr)^{1/2}
\ee
and
\be\label{gAdef}
g_{yAx}(z):= \frac{1}{\|\tilde g_{yz}\|_A} 1_{z\in A} \,1_{x\in A} \, \tilde g_{yx}(z)\,.
\ee
Some functions of this type are depicted in Figure~\ref{fig:gAplot}.

\begin{figure}[h]
\begin{center}
\includegraphics[width=7cm]{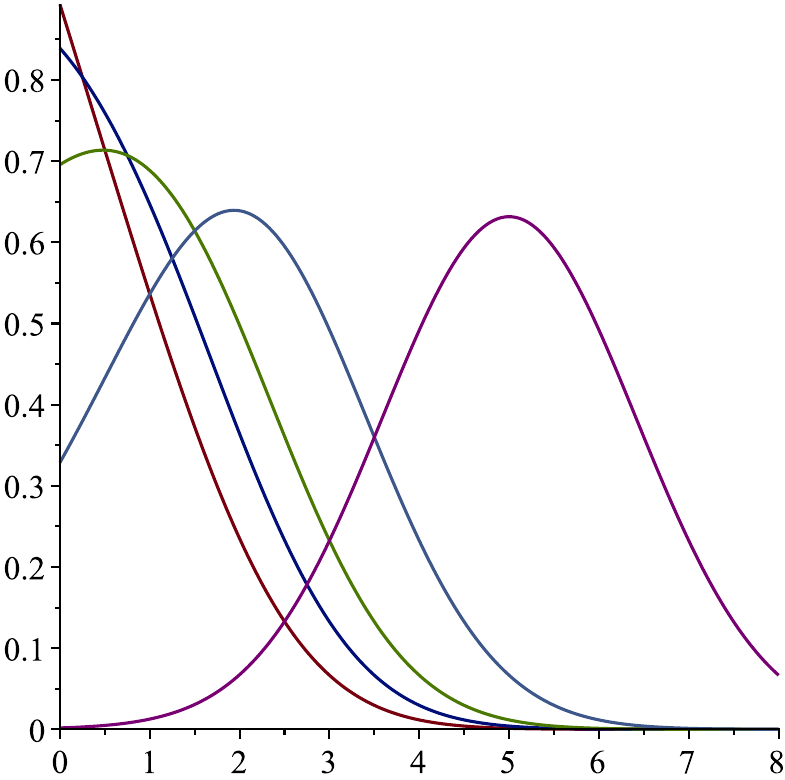}
\end{center}
\caption{Examples of functions of the type \eqref{gAdef}, but defined on the real line instead of a hyperboloid, here with $A=[0,\infty)$. The first factor on the right-hand side of \eqref{gAdef} causes a deviation from the Gaussian shape which is small for $x$ far from the boundary of $A$ but visible for $x$ close to it. The right axis shown is $z$, the functions plotted are $g_{Ax}(z)= \|g_z \|_A^{-1} \, \tilde g_x(z)$ with $\tilde g_x$ the Gaussian density with center $x$ and width 1 for the values $x=0,\tfrac{1}{2},1,2,5$ (in the order of centers from left to right, or of decreasing values of $g_{Ax}(0)$).}
\label{fig:gAplot}
\end{figure}

Since
\be
\tilde g_{yx}(z) = \tilde g_{yz}(x)\,,
\ee
it follows that for every $z\in \HHH_y(x)$,
\begin{subequations}
\begin{align}
\int_A d^3x \, g_{yAx}(z)^2
&=  1_{z\in A} \int_A d^3x \,  \frac{1}{\|\tilde g_{yz}\|_A^2}\tilde g_{yx}(z)^2 \\
&=   \frac{1_{z\in A}}{\|\tilde g_{yz}\|_A^2} \int_A d^3x \, \tilde g_{yx}(z)^2 \\
&=   \frac{1_{z\in A}}{\|\tilde g_{yz}\|_A^2} \int_A d^3x \, \tilde g_{yz}(x)^2 \\
&=   \frac{1_{z\in A}}{\|\tilde g_{yz}\|_A^2} \, \| \tilde g_{yz}\|^2 \\
&= 1_{z\in A}\,.\label{gyAxnormalized}
\end{align}
\end{subequations}
Again, we write $g_{yAxi}$ for the function $g_{yAx}$ applied to the $i$-th variable. It follows from \eqref{gyAxnormalized} that, for any $i\in\{1,\ldots,N\}$ and $A\subseteq \HHH_y(s)$,
\be\label{Gnormalized}
\int_A d^3x \, P_{\HHH_y(s)}(g_{yAxi})^2 = P_{\HHH_y(s)}\bigl(\HHH_y(s)^{i-1}\times A \times \HHH_y(s)^{N-i-1}\bigr)\,.
\ee
We define
\be\label{Ldef}
L(x_1,x_2) := 
\begin{cases}
U^0_{\HHH_2} \, P_{\HHH_2}(g_{y_2P_2x_22}) \, U^{\HHH_2}_{\HHH_1} \, P_{\HHH_1}(g_{y_1P_1x_11}) \, U^{\HHH_1}_0 & \text{if }x_1\in P_1, x_2\in P_2\\
U^0_{\HHH_2} \, P_{\HHH_2}(g_{y_2F_2x_22}) \, U^{\HHH_2}_{\HHH_1} \,P_{\HHH_1}(g_{y_1P_1x_11}) \, U^{\HHH_1}_0 & \text{if }x_1\in P_1, x_2\in F_2\\
U^0_{\HHH_1} \, P_{\HHH_1}(g_{y_1F_1x_11})  \, U^{\HHH_1}_{\HHH_2} \, P_{\HHH_2}(g_{y_2P_2x_22}) \, U^{\HHH_2}_0 &\text{if }x_1\in F_1,x_2\in P_2\\
U^0_{\HHH_1} \, P_{\HHH_1}(g_{y_1F_1x_11}) \, U^{\HHH_1}_{\HHH_2} \, P_{\HHH_2}(g_{y_2F_2x_22}) \, U^{\HHH_2}_0 &\text{if }x_1\in F_1, x_2 \in F_2.
\end{cases}
\ee

\begin{prop}\label{prop:IN=2} 
Interaction locality implies that
\be
\int d^4x_1 \int d^4x_2 \, D(x_1,x_2) = I\,.
\ee
As a consequence, $G(\cdot)$ is a POVM, and \eqref{PPPdef2} defines a probability distribution for every $\psi_0\in \Hilbert_0$ with $\|\psi_0\|=1$.
\end{prop}

\begin{proof}
Since (coarea formula)
\be\label{coarea}
\int\limits_{\future(y)} \hspace{-3mm} d^4x\, f(x,y) = \int\limits_0^\infty ds \int\limits_{\HHH_y(s)} \hspace{-1mm} d^3x \, f(x,y)\,,
\ee 
and since
\be\label{expint}
\int_0^\infty \hspace{-3mm} ds \, \tau^{-1} \exp(-s/\tau)=1\,, 
\ee
it suffices to show that for any two hyperboloids $\HHH_1,\HHH_2$ based at $y_1$ and $y_2$, respectively,
\be\label{toshow}
\int_{\HHH_1}d^3x_1 \int_{\HHH_2} d^3x_2 \, L(x_1,x_2)^\dagger \, L(x_1,x_2) = I\,.
\ee
Let $\Sigma:= P_1\cup P_2$. Writing $\HHH_i = P_i \cup F_i$ yields four parts for $\HHH_1\times \HHH_2$. We deal with each part separately, beginning with $P_1\times P_2$: By the consequence \eqref{ILB} of IL, 
\be\label{move1}
U_{\HHH_2}^{\Sigma} \, P_{\HHH_2}(\HHH_2 \times P_2) \, U_{\Sigma}^{\HHH_2} =  P_{\Sigma}(\Sigma \times P_2)\,.
\ee
By the consequence \eqref{ILf} of IL,
\be\label{move2}
U_{\HHH_1}^{\Sigma} \, P_{\HHH_1}(g_{y_1P_1x_11}) \, U_{\Sigma}^{\HHH_1} = P_{\Sigma}(g_{y_1P_1x_11})
\ee
for every $x_1\in P_1$. Thus,
\begin{subequations}\label{simplify1}
\begin{align}
&\int_{P_1}d^3x_1 \int_{P_2} d^3x_2 \, L(x_1,x_2)^\dagger \, L(x_1,x_2)\\
=~~& \int_{P_1}d^3x_1 \int_{P_2} d^3x_2 \, U^0_{\HHH_1} \, P_{\HHH_1}(g_{y_1P_1x_11}) \, U^{\HHH_1}_{\HHH_2} \, P_{\HHH_2}(g_{y_2P_2x_22})^2 \, U^{\HHH_2}_{\HHH_1} \, P_{\HHH_1}(g_{y_1P_1x_11}) \, U^{\HHH_1}_0\\
\stackrel{\eqref{Gnormalized}}{=}~& \int_{P_1}d^3x_1  \, U^0_{\HHH_1} \, P_{\HHH_1}(g_{y_1P_1x_11}) \, U^{\HHH_1}_{\HHH_2} \, P_{\HHH_2}(\HHH_2 \times P_2) \, U^{\HHH_2}_{\HHH_1} \, P_{\HHH_1}(g_{y_1P_1x_11}) \, U^{\HHH_1}_0\\
\stackrel{\eqref{move1},\eqref{move2}}{=}& \int_{P_1}d^3x_1  \, U^0_{\Sigma} \, P_\Sigma(g_{y_1P_1x_11}) \, P_{\Sigma}(\Sigma \times P_2) \,  P_\Sigma(g_{y_1P_1x_11}) \, U^{\Sigma}_0\\
=~~& \int_{P_1}d^3x_1  \, U^0_{\Sigma} \, P_{\Sigma}(g_{y_1P_1x_11})^2 \,  P_{\Sigma}(\Sigma \times P_2) \, U^{\Sigma}_0\label{simplify1commute}\\
\stackrel{\eqref{Gnormalized}}{=}~& U^0_{\Sigma} \, P_{\Sigma}(P_1\times \Sigma) \,  P_{\Sigma}(\Sigma\times P_2) \, U^{\Sigma}_0\\[3mm]
=~~& U^0_{\Sigma} \, P_{\Sigma}(P_1 \times P_2) \, U^{\Sigma}_0 \,.\label{P1P2}
\end{align}
\end{subequations}
Here, we used in \eqref{simplify1commute} that the operators of a PVM commute, and in the last step that, for every PVM, $P(A)P(B)=P(A\cap B)$.

We now turn to the contribution from $P_1\times F_2$. Here we exploit that, by the consequence \eqref{ILB} of IL applied to $\HHH_2$ and $\Sigma=P_1\cup P_2$,
\be\label{move3}
P_{\Sigma}(\Sigma \times P_1) = U^{\Sigma}_{\HHH_2} \, P_{\HHH_2}(\HHH_2 \times F_2) \, U^{\HHH_2}_{\Sigma} \,.
\ee
With the same strategy as in \eqref{simplify1}, we now obtain that
\be\label{P1F2}
\int_{P_1}d^3x_1 \int_{F_2} d^3x_2 \, L(x_1,x_2)^\dagger \, L(x_1,x_2)
= U^0_{\Sigma} \, P_{\Sigma}(P_1 \times P_1) \, U^{\Sigma}_0\,.
\ee
For $F_1\times P_2$, we interchange the order of integration so that the $x_1$ integration is carried out first (i.e., inside the $x_2$ integral). Exploiting that, by \eqref{ILB},
\be\label{move4}
P_{\Sigma}(P_2\times \Sigma) = U_{\HHH_1}^{\Sigma} \, P_{\HHH_1}(F_1\times \HHH_1) \, U_{\Sigma}^{\HHH_1}  \,,
\ee
we obtain through the same strategy as before that
\be\label{F1P2}
\int_{P_2} d^3x_2 \int_{F_1}d^3x_1 \, L(x_1,x_2)^\dagger \, L(x_1,x_2)
= U^0_{\Sigma} \, P_{\Sigma}(P_2 \times P_2) \, U^{\Sigma}_0\,.
\ee
Likewise for $F_1\times F_2$:
\be\label{F1F2}
\int_{F_1}d^3x_1 \int_{F_2} d^3x_2 \, L(x_1,x_2)^\dagger \, L(x_1,x_2)
= U^0_{\Sigma} \, P_{\Sigma}(P_2 \times P_1) \, U^{\Sigma}_0\,.
\ee
Putting together \eqref{P1P2}, \eqref{P1F2}, \eqref{F1P2}, and \eqref{F1F2}, we obtain that
\begin{subequations}
\begin{align}
&\int_{\HHH_1}d^3x_1 \int_{\HHH_2} d^3x_2 \, L(x_1,x_2)^\dagger \, L(x_1,x_2) \\
&= U^0_{\Sigma} \, P_{\Sigma}\Bigl( (P_1 \times P_2)\cup (P_1\times P_1) \cup (P_2 \times P_2) \cup (P_2\times P_1) \Bigr) \, U^{\Sigma}_0\\[2mm]
&= U^0_{\Sigma} \, P_{\Sigma}(\Sigma^2) \, U^{\Sigma}_0
= U^0_{\Sigma} \, I \, U^{\Sigma}_0=I\,,
\end{align}
\end{subequations}
as claimed in \eqref{toshow}.
\end{proof}

\subsection{General Case}

Consider $n_i$ flashes for particle $i$; they occur at the random points $X_{ik}$, $i\in\{1,\ldots,N\}$, $k\in \{1,\ldots,n_i\}$. Let $\uX$ denote again the collection of all $X_{ik}$ with $1\leq i\leq N$ and $1\leq k\leq n_i$, likewise $\ux$ the collection of the space-time points $x_{ik}$, and $d\ux$ as in \eqref{dux}. For each $i$, let $x_{i0}$ be a given seed flash. The distribution of $\uX$ is again of the form
\be\label{PPPdef}
\PPP\bigl(\uX \in d\ux \bigr) = \scp{\psi_0}{D(\ux)|\psi_0} \, d\ux
\ee
with $D$ again of the form
\be\label{Ddef}
D(\ux) = \biggl( \frac{1}{\tau^\nu} \prod_{i=1}^N \prod_{k=1}^{n_i} 1_{x_{ik}\in \future(x_{ik-1})} e^{-|x_{ik}-x_{ik-1}|/\tau} \biggr)  L(\ux)^\dagger \, L(\ux)\,.
\ee
In particular, the distribution of $\uX$ is again determined by a POVM $G(d\ux) = D(\ux) \, d\ux$ with density $D(\ux)$. We use the notation $\HHH_{ik}$ as in \eqref{HHHikdef}. The first, rough idea would be to take $L$ to be something like
\be
\text{``}~~L(\ux) = \prod_{i=1}^N \prod_{k=1}^{n_i} U^0_{\HHH_{ik}} \, P_{\HHH_{ik}}(g_{x_{ik-1}x_{ik}i}) \, U^{\HHH_{ik}}_0~~\text{''}
\ee
with a problem of operator ordering. To address this problem, we need to construct the analogs of the 3-cells $P_i$ and $F_i$ of the previous section.

\subsubsection{Division into Cells}

The connected components of $\MMM\setminus\cup_{ik}\HHH_{ik}$ (more precisely, their closures) we call \emph{4-cells}. They can be labeled by $\uk=(k_1,\ldots,k_N) \in \prod_{i=1}^N \{0,\ldots,n_i\}$: the 4-cell for $\uk$ is defined as
\be\label{4Cdef}
{}^{4}C_{\uk} := \bigcap_{i=1}^N \Bigl(\future(\HHH_{ik_i}) \cap \past(\HHH_{ik_i+1}) \Bigr)\,,
\ee
where $\future(\HHH_{i0})$ and $\past(\HHH_{in_i+1})$ should be understood as $\MMM$; see Figure~\ref{fig:4cellsN=2} for an example. There are $\prod_i (n_i+1)$ 4-cells. The 4-cells form a partition of space-time, except for overlap on the hyperboloids.
For $\uk$ with all $k_i=0$ we write $0^N$, and we write $\un=(n_1,\ldots,n_N)$, as well as ${}^4\sC$ for the set of all 4-cells.

\begin{figure}[h]
\begin{center}
\begin{tikzpicture}
\draw [very thick, domain=-2.5:2.5, samples=50] plot (\x, {-0.5+sqrt(1+(\x+0.8)*(\x+0.8)))} );
\node at (-2.85,1.65) {$\HHH_{11}$};
\draw [very thick, domain=-2.5:2.5, samples=50] plot (\x, {-0.3+sqrt(0.1+(\x-0.5)*(\x-0.5))} );
\node at (2.9,1.8) {$\HHH_{21}$};
\node at (0,-0.4) {${}^4C_{00}$};
\node at (1.1,1) {${}^4C_{01}$};
\node at (-1.55,1.2) {${}^4C_{10}$};
\node at (0,2.2) {${}^4C_{11}$};
\end{tikzpicture}
\end{center}
\caption{Notation for 4-cells as in \eqref{4Cdef}}
\label{fig:4cellsN=2}
\end{figure}
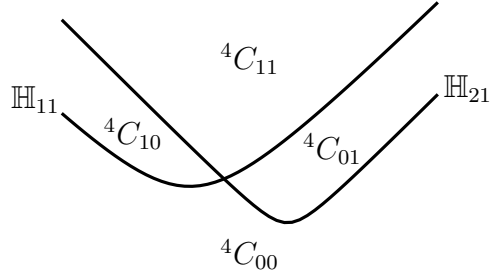

The faces of the 4-cells are pieces of hyperboloids henceforth called \emph{3-cells},
\be
{}^3C_{i\uk} := \HHH_{ik_i}\cap \bigcap_{j\neq i} \Bigl(\future(\HHH_{jk_j}) \cap \past(\HHH_{jk_j+1})  \Bigr)\,,
\ee
where $\uk$ must be such that $k_i\geq 1$.
In fact, ${}^3C_{i\uk}$ is the common boundary of ${}^4C_{\uk}$ and ${}^4C_{\uk'}$, where $k'_i=k_i-1$ and $k'_j=k_j$ for all $j\neq i$.
For example, for two hyperboloids as in Figures~\ref{fig:PFH} and \ref{fig:4cellsN=2}, $P_1= {}^3C_{110}$, $P_2={}^3C_{201}$, and $F_i= {}^3C_{i11}$. The set of all 3-cells will be denoted by ${}^3\sC$.

If two 4-cells border on each other along a 3-cell ${}^3C_{i\uk}$, then the one in the future of $\HHH_{ik_i}$ will henceforth be said to be a \emph{successor} of the one in the past of $\HHH_{ik_i}$, and conversely a \emph{predecessor}. The predecessors of ${}^4C_{\uk}$ are those for which one $k_j$ in $\uk$ has been replaced by $k_j-1$.

We say that a set $S\subseteq \MMM$ is \emph{past complete} if $\past(S)\subseteq S$; correspondingly \emph{future complete}. For example, $\emptyset$ and $\MMM$ are both past and future complete, the past of any set is past complete, an intersection of past complete sets is past complete, ${}^4C_{0^N}$ is past complete, ${}^4C_{\un}$ is future complete, and $\MMM\setminus {}^4C_{\un}$ is past complete. One easily verifies that the complement of a past complete set is future complete and vice versa.

\begin{prop}\label{prop:pastcomplete}
Every (closed) past complete set $S$ except $\emptyset$ and $\MMM$ is the past of its boundary, $S=\past(\partial S)$, and $\partial S$ is a spacelike-or-lightlike hypersurface.
\end{prop}

\begin{proof}
For any $x\in S$, consider a timelike straight line (geodesic) $\gamma$ through $x$; there must be a point on $\gamma$ outside $S$, or else $S=\MMM$ by past completeness. Again by past completeness, $\gamma$ must lie in $S$ up to a point $\gamma(s_0)$ and outside from there onwards. So $\gamma(s_0)$ must lie on $\partial S$, and $x\in\past(\gamma(s_0))\subseteq \past(\partial S)$.
\end{proof}

While the exact location and shape of ${}^4C_{\uk}$ depends on the hyperboloids, many relations between the 4-cells, such as which one borders on which others along which 3-cells, can be read off from the index $\uk$. That is why we also call $\uk\in \prod_i\{0,\ldots,n_i\}$ an \emph{abstract 4-cell} and a pair $(i,\uk)$ such that $k_i\geq 1$ an \emph{abstract 3-cell}. The set of abstract 4-cells (respectively, 3-cells) is ${}^4\sA:=\prod_i\{0,\ldots,n_i\}$, respectively ${}^3\sA :=\{(i,\uk)\in \{1\ldots N\}\times {}^4\sA:k_i\geq 1\}$. The \emph{future faces} of the abstract 4-cell $\uk$ are the abstract 3-cells $(i,\uk')$ with $k'_i=k_i+1$ and $k'_j=k_j$ for all $j\neq i$ (if they exist); the \emph{past faces} of $\uk$ are the $(i,\uk)$ with $k_i\geq 1$. The \emph{predecessors} of $\uk$ are those abstract 4-cells for which one $k_j$ in $\uk$ has been replaced by $k_j-1$; correspondingly \emph{successors}. A set $V$ of abstract 4-cells will be called \emph{predecessor complete} iff\footnote{iff = if and only if} it contains every predecessor of each of its elements; correspondingly \emph{successor complete}. A set is successor complete iff its complement is predecessor complete.

\begin{prop}
If a set $V\subseteq {}^4\sA$ is predecessor complete, then the corresponding space-time set $S(V) = \cup_{\uk \in V} {}^4C_{\uk}$ is past complete. Furthermore, if the hyperboloids are such that ${}^4C_{\uk}$ has non-empty interior (or non-zero 4-volume) for every $\uk\in {}^4\sA$, then also the converse is true: $S(V)$ is past complete only if $V$ is predecessor complete.
\end{prop}

\begin{proof}
To see that $S(V)$ is past complete, consider $x\in S(V)$ and $y$ in the past of $x$. The straight line (or any causal curve) from $x$ to $y$, when crossing hyperboloids, enters a predecessor of the 4-cell, and thus remains in $S(V)$. We remark that if some ${}^4C_{\uk}$ is empty (as would happen if $x_{ik_i-1} \in \future(x_{jk_j+1})$ with $j\neq i$), then $S(\{\uk\}) = {}^4C_{\uk}=\emptyset$ is past complete although $V=\{\uk\}$ is not predecessor complete.

Now assuming that the 4-cells have non-empty interior, if $S(V)$ were past complete but $V$ not predecessor complete, then let $\uk'\notin V$ be a predecessor of $\uk\in V$. Since the interiors are non-empty, there are interior points $x\in {}^4C_{\uk}$ and $y\in {}^4C_{\uk'}$ such that $y\in \past(x)$, in contradiction to $y\in S(V)$.
\end{proof}

Let $\sN$ be the set of predecessor complete sets of abstract 4-cells. It becomes a directed network by putting a directed edge from $V_1$ to $V_2$ whenever $V_2$ can be obtained from $V_1$ by adding one abstract 4-cell, $V_2=V_1\cup\{\uk\}$. In particular, every edge is related to some abstract 4-cell, while the same abstract 4-cell can occur for several edges at different vertices. An \emph{admissible sequence} $(V_1,\ldots,V_{r+1})$ is a path in $\sN$ (using only edges in their direction) from the vertex $\emptyset$ to the vertex ${}^4\sA$. We will show in Proposition~\ref{prop:existsadmseq} that admissible sequences exist.

\begin{figure}[h]
\begin{center}
\begin{tikzpicture}
\filldraw (0,1.41) circle [radius=0.07];
\filldraw (0,0) circle [radius=0.07];
\filldraw (1,-1) circle [radius=0.07];
\filldraw (-1,-1) circle [radius=0.07];
\filldraw (0,-2) circle [radius=0.07];
\filldraw (0,-3.41) circle [radius=0.07];
\draw[->] (0,-2) -- (-0.5,-1.5);
\draw[->] (0,-2) -- (0.5,-1.5);
\draw[->] (-1,-1) -- (-0.5,-0.5);
\draw[->] (1,-1) -- (0.5,-0.5);
\draw[->] (0,0) -- (0,0.71);
\draw[->] (0,-3.41) -- (0,-2.7);
\draw (-0.5,-1.5) -- (-1,-1) (0.5,-1.5) -- (1,-1) (-0.5,-0.5) -- (0,0) (0.5,-0.5) -- (0,0) (0,0.71) -- (0,1.41) (0,-2.7) -- (0,-2);
\node at (0.6,-2.1) {$\{00\}$};
\node at (1.9,-1) {$\{00,01\}$};
\node at (-1.9,-1) {$\{00,10\}$};
\node at (1.2,0.1) {$\{00,01,10\}$};
\node at (0,1.8) {$\{00,01,10,11\}={}^4\sA$};
\node at (0,-3.7) {$\emptyset$};
\end{tikzpicture}
\end{center}
\caption{The directed network $\sN$ for two hyperboloids as in Figure~\ref{fig:4cellsN=2}. There are two paths from $\emptyset$ to ${}^4\sA$, both of which are admissible sequences.}
\label{fig:network}
\end{figure}
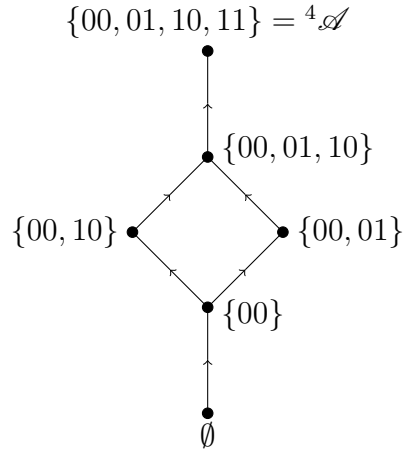

\begin{figure}[h]
\begin{center}
\begin{minipage}{3.5cm}
\begin{tikzpicture}[scale=0.7]
\draw[fill=gray, draw=gray] (-2.5,-0.5) --plot[domain=-2.5:-0.342, samples=30] (\x, {-0.5+sqrt(1+(\x+0.8)*(\x+0.8))} ) -- (-0.342,-0.5) -- cycle;
\draw[fill=gray, draw=gray] (-0.342,-0.5) --plot[domain=-0.342:2.5, samples=30] (\x, {-0.3+sqrt(0.1+(\x-0.5)*(\x-0.5))} ) -- (2.5,-0.5) -- cycle;
\draw [very thick, domain=-2.5:2.5, samples=50] plot (\x, {-0.5+sqrt(1+(\x+0.8)*(\x+0.8)))} );
\draw [very thick, domain=-2.5:2.5, samples=50] plot (\x, {-0.3+sqrt(0.1+(\x-0.5)*(\x-0.5))} );
\end{tikzpicture}\\[-8mm]
\begin{center}
$\{00\}$
\end{center}
\end{minipage}~~
\begin{minipage}{3.5cm}
\begin{tikzpicture}[scale=0.7]
\draw[fill=gray, draw=gray] (-2.5,-0.5) --plot[domain=-2.5:-0.342, samples=30] (\x, {-0.3+sqrt(0.1+(\x-0.5)*(\x-0.5))} ) -- (-0.342,-0.5) -- cycle;
\draw[fill=gray, draw=gray] (-0.342,-0.5) --plot[domain=-0.342:2.5, samples=30] (\x, {-0.3+sqrt(0.1+(\x-0.5)*(\x-0.5))} ) -- (2.5,-0.5) -- cycle;
\draw [very thick, domain=-2.5:2.5, samples=50] plot (\x, {-0.5+sqrt(1+(\x+0.8)*(\x+0.8)))} );
\draw [very thick, domain=-2.5:2.5, samples=50] plot (\x, {-0.3+sqrt(0.1+(\x-0.5)*(\x-0.5))} );
\end{tikzpicture}\\[-8mm]
\begin{center}
$\{00,10\}$
\end{center}
\end{minipage}~~
\begin{minipage}{3.5cm}
\begin{tikzpicture}[scale=0.7]
\draw[fill=gray, draw=gray] (-2.5,-0.5) --plot[domain=-2.5:-0.342, samples=30] (\x, {-0.5+sqrt(1+(\x+0.8)*(\x+0.8))} ) -- (-0.342,-0.5) -- cycle;
\draw[fill=gray, draw=gray] (-0.342,-0.5) --plot[domain=-0.342:2.5, samples=30] (\x, {-0.5+sqrt(1+(\x+0.8)*(\x+0.8))} ) -- (2.5,-0.5) -- cycle;
\draw [very thick, domain=-2.5:2.5, samples=50] plot (\x, {-0.5+sqrt(1+(\x+0.8)*(\x+0.8)))} );
\draw [very thick, domain=-2.5:2.5, samples=50] plot (\x, {-0.3+sqrt(0.1+(\x-0.5)*(\x-0.5))} );
\end{tikzpicture}\\[-8mm]
\begin{center}
$\{00,01\}$
\end{center}
\end{minipage}~~
\begin{minipage}{3.5cm}
\begin{tikzpicture}[scale=0.7]
\draw[fill=gray, draw=gray] (-2.5,-0.5) --plot[domain=-2.5:-0.342, samples=30] (\x, {-0.3+sqrt(0.1+(\x-0.5)*(\x-0.5))} ) -- (-0.342,-0.5) -- cycle;
\draw[fill=gray, draw=gray] (-0.342,-0.5) --plot[domain=-0.342:2.5, samples=30] (\x, {-0.5+sqrt(1+(\x+0.8)*(\x+0.8))} ) -- (2.5,-0.5) -- cycle;
\draw [very thick, domain=-2.5:2.5, samples=50] plot (\x, {-0.5+sqrt(1+(\x+0.8)*(\x+0.8)))} );
\draw [very thick, domain=-2.5:2.5, samples=50] plot (\x, {-0.3+sqrt(0.1+(\x-0.5)*(\x-0.5))} );
\end{tikzpicture}\\[-8mm]
\begin{center}
$\{00,01,10\}$
\end{center}
\end{minipage}
\end{center}
\caption{The space-time sets (unions of 4-cells) $S(V)$ corresponding to some vertices $V$ in $\sN$ for two hyperboloids}
\label{fig:network4cells}
\end{figure}
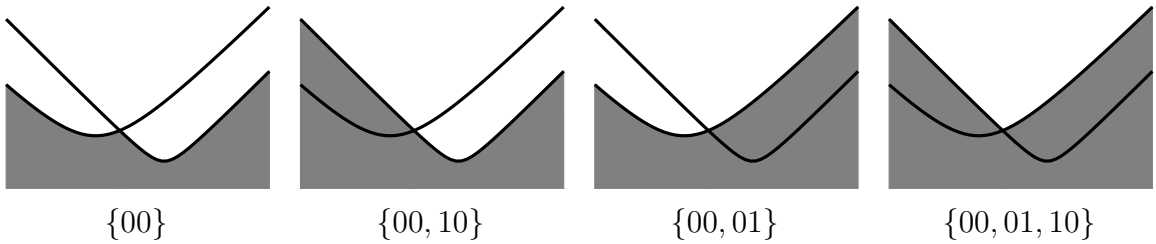

We say that the admissible sequence $(V_1,\ldots,V_{r+1})$ \emph{crosses the 4-cell $\uk$ in step $n$} iff $V_{n+1}=V_n\cup \{\uk\}$. We say that it \emph{crosses the 3-cell $(i,\uk)\in {}^3\sA$ in step $n$} iff $V_{n+1}=V_n\cup \{\uk\}$.

\begin{prop}\label{prop:cross}
Every admissible sequence crosses every 4-cell and every 3-cell exactly once.
\end{prop}

\begin{proof}
Since each step in the path adds exactly one 4-cell, and since the last element of the sequence is the set of all 4-cells, each 4-cell must occur sooner or later, and cannot occur twice. The 3-cell $(i,\uk)$ gets crossed exactly when the 4-cell $\uk$ gets crossed.
\end{proof}

In particular, $r$ equals the number of 4-cells. Since the starting point is fixed, an admissible sequence can be characterized by specifying which edge to use in each step. Since the edges are labeled with abstract 4-cells, it can be specified by the sequence $(\uk_1,\ldots,\uk_r)$ of abstract 4-cells in the order in which they are crossed. Such a sequence is an ordering of the set of all abstract 4-cells. However, not every ordering of the 4-cells corresponds to an admissible sequence.

\begin{prop}\label{prop:admseq}
An ordering $(\uk_1,\ldots,\uk_r)$ of the 4-cells corresponds to an admissible sequence iff for every $n\in\{1,\ldots,r\}$, every predecessor of $\uk_n$ occurred earlier. 
\end{prop}

\begin{proof}
``only if'': Otherwise $V_{n+1}=V_n\cup \{\uk_n\}$ is not predecessor complete, as $V_n=\{\uk_1,\ldots,\uk_{n-1}\}$.

``if'': The sequence of 4-cells tells us in each step of the path in $\sN$ which edge to take. In order to verify that such edges exist in $\sN$, we need to check that, for each step from $V_n$ to $V_{n+1}=V_n\cup \{\uk_n\}$, $V_{n+1}$ is predecessor complete. It is because each predecessor of $\uk_n$ is contained in $V_n=\{\uk_1,\ldots,\uk_{n-1}\}$ by assumption, and because $V_n$ is predecessor complete. Since every 4-cell occurs, the end point of the path in $\sN$ is the set of all 4-cells.
\end{proof} 

As a consequence, the sequence of 4-cells must begin with $0^N$ (the only one without predecessor) and end with $\un$ (the only one without successor). For $N=2$ and $n_1=1=n_2$ as in Figures~\ref{fig:4cellsN=2}, \ref{fig:network}, and \ref{fig:network4cells}, there are two orderings as described in Proposition~\ref{prop:admseq}: (00,01,10,11) and (00,10,01,11).

\begin{prop}\label{prop:existsadmseq}
For every choice of $N,n_1,\ldots,n_N\in\NNN$, there exists an admissible sequence.
\end{prop}

\begin{proof}
For every $\uk$, define $m(\uk)= k_1+\ldots+k_N$. 
We specify the ordering of 4-cells. Begin with $\uk=0^N$, the only 4-cell with $m(\uk)=0$. Then list, in arbitrary order, all 4-cells $\uk$ with $m(\uk)=1$. Then, in arbitrary order, all 4-cells $\uk$ with $m(\uk)=2$, and so on up to $m(\uk)=\nu$, which occurs only for $\uk=\un$. Then all 4-cells have occurred exactly once. Every predecessor $\uk'$ of $\uk$ occurred earlier than $\uk$ because $m(\uk')=m(\uk)-1$. (We remark that not every admissible sequence needs to have this structure.)
\end{proof}

Of two admissible sequences, we say that they differ by an \emph{elementary deformation} if their associated orderings of 4-cells differ only by an exchange of two successive 4-cells, i.e., one is $(\uk_1,\ldots,\uk_r)$ and the other
\be\label{deform}
\bigl(\uk_1,\ldots,\uk_{n-1},\uk_{n+1},\uk_n,\uk_{n+2},\ldots,\uk_r\bigr)
\ee
for some $n\in \{1,\ldots,r-1\}$. For example, in Figure~\ref{fig:network} this exchange corresponds to switching from the left path to the right one or vice versa. An exchange of 4-cells as in \eqref{deform}, applied to an admissible sequence, does not necessarily yield another admissible sequence, but here we need the converse fact:

\begin{prop}\label{prop:deform}
Any two admissible sequences can be obtained from each other through finitely many elementary deformations.
\end{prop}

\begin{proof}
Let $(\uk_1,\ldots,\uk_r)$ be the ordering of 4-cells corresponding to one of the two admissible sequences and $(\uk'_1,\ldots,\uk'_r)$ the other. Apply the following elementary deformations to the primed ordering. Find the place where $\uk_1$ occurs and move $\uk_1$ one place to the left in the primed ordering (by exchange with its left neighbor). The resulting ordering corresponds to an admissible sequence because $\uk_1$ has no predecessor. Likewise, $\uk_1$ can be moved again to the left, in fact repeatedly until it reaches the first position. Repeating the procedure, we can move $\uk_2$ to the second position and so on until we have reached the unprimed ordering. In each intermediate ordering, predecessors always occur earlier, because they did in the two given orderings.
\end{proof}

\subsubsection{Definition of $L$}

We use an admissible sequence to define the operator ordering in $L(\ux)$, and then proceed to show that the operator $L(\ux)$ does not, in fact, depend on the choice of admissible sequence.

So fix an admissible sequence. Since each $\HHH_{ik}$ is partitioned into 3-cells, there is exactly one 3-cell ${}^3C(x_{ik})$ containing $x_{ik}$ (except in the probability-0 case that $x_{ik}$ lies on the boundary between two 3-cells on $\HHH_{ik}$, which we ignore). To the flash $x_{ik}$ we associate the operator
\be\label{Kdef}
K(x_{ik}):= U^0_{\HHH_{ik}} \, P_{\HHH_{ik}}\bigl( g_{x_{ik-1},{}^3C(x_{ik}),x_{ik},i} \bigr) \, U^{\HHH_{ik}}_0 \,.
\ee  
We define $L(\ux)$ as the product of the $K(x_{ik})$ in the order from right to left in which the 3-cells are crossed in the admissible sequence. Now in some steps of the sequence, several 3-cells are crossed in the same step. Among these, it does not matter which order we choose, as their operators commute:

\begin{prop}\label{prop:3Ccommute}
Assume interaction locality and consider $V$ and $V'=V\cup \{\uk\}$ in $\sN$. If ${}^3C(x_{ik})$ and ${}^3C(x_{j\ell})$ are two 3-cells in the common boundary of $S(V)$ and ${}^4C_{\uk}$ (so $k=k_i$ and $\ell=k_j$), then $K(x_{ik})$ commutes with $K(x_{j\ell})$. As a consequence, every admissible sequence unambiguously defines a product $L(\ux)$.
\end{prop}

\begin{proof}
Since ${}^3C(x_{ik})\subseteq \HHH_{ik} \cap \partial S(V)$, and since $g_{x_{ik-1},{}^3C(x_{ik}),x_{ik}}$ vanishes outside of ${}^3C(x_{ik})$, the consequence \eqref{ILf} of interaction locality implies that multiplication by this $g$ function can as well be carried out on $\partial S(V)$, i.e.,
\begin{subequations}
\begin{align}
K(x_{ik}) \, K(x_{j\ell})
&= U^0_{\HHH_{ik}} \, P_{\HHH_{ik}}\bigl( g_{x_{ik-1},{}^3C(x_{ik}),x_{ik},i} \bigr) \, U^{\HHH_{ik}}_{\HHH_{j\ell}} \, P_{\HHH_{j\ell}}\bigl( g_{x_{j\ell-1},{}^3C(x_{j\ell}),x_{j\ell},j} \bigr) \, U^{\HHH_{j\ell}}_0\\[1mm]
&= U^0_{\partial S(V)} \, P_{\partial S(V)}\bigl( g_{x_{ik-1},{}^3C(x_{ik}),x_{ik},i} \bigr) \, P_{\partial S(V)}\bigl( g_{x_{j\ell-1},{}^3C(x_{j\ell}),x_{j\ell},j} \bigr) \, U^{\partial S(V)}_0\\
&= U^0_{\partial S(V)} \, P_{\partial S(V)}\Bigl( g_{x_{ik-1},{}^3C(x_{ik}),x_{ik},i} \, g_{x_{j\ell-1},{}^3C(x_{j\ell}),x_{j\ell},j} \Bigr) \, U^{\partial S(V)}_0\,.
\end{align}
\end{subequations}
Since multiplication of the two $g$ functions is commutative, $K(x_{j\ell})\, K(x_{ik})$ yields the same expression.
\end{proof}

\begin{prop}\label{prop:L}
Assuming interaction locality, any two admissible sequences lead to the same operator $L(\ux)$.
\end{prop}

\begin{proof}
By Proposition~\ref{prop:deform}, it suffices to consider two admissible sequences that differ by an elementary deformation as in \eqref{deform}. By Proposition~\ref{prop:admseq}, the two 4-cells $\uk_n,\uk_{n+1}$ that get exchanged must be such that neither is a predecessor of the other; that is, they do not have a 3-cell in common. Hence, for each of them the past boundary is a subset of $\partial S(V_n)$ with $V_n=\{\uk_1, \ldots , \uk_{n-1}\}$. For the same reasons as in the proof of Proposition~\ref{prop:3Ccommute}, the $K$ operators for any two flashes in 3-cells in the past boundaries of ${}^4C_{\uk_n}$ and ${}^4C_{\uk_{n+1}}$ commute (they are multiplication operators on a common spacelike surface). Thus, the different operator orderings associated with the two admissible sequences yield the same $L(\ux)$.
\end{proof}

This completes the definition of $L(\ux)$ and thus of $D(\ux)$ as in \eqref{Ddef} and of the distribution of $\uX$ as in \eqref{PPPdef}. It remains to verify that $\PPP$ is a probability distribution.

\subsubsection{Normalization}

\begin{prop}\label{prop:normalized}
Interaction locality implies that
\be\label{normalized}
\int\limits_{\MMM^\nu} \hspace{-1mm} d\ux\, D(\ux) = I\,.
\ee
As a consequence, $G(\cdot)$ is a POVM, and \eqref{PPPdef} defines a probability distribution for every $\psi_0\in\Hilbert_0$ with $\|\psi_0\|=1$.
\end{prop}

\begin{proof}
Written out, \eqref{normalized} reads
\be
\frac{1}{\tau^\nu}\Biggl(\prod_{ik}\int\limits_{\future(x_{ik-1})} \hspace{-6mm} d^4x_{ik}\Biggr) \Biggl( \prod_{ik} e^{-|x_{ik}-x_{ik-1}|/\tau} \Biggr)  L(\ux)^\dagger L(\ux)=I
\ee
with the abuse of notation that $\prod_{ik} \int d^4x_{ik}$ means, not a product, but repeated integration over all $x_{ik}$, with each integral extending up to the equal sign.

By the coarea formula \eqref{coarea} and the normalization \eqref{expint}, it suffices to show that for all $s_{ik}>0$,
\be\label{toshow2}
\Biggl(\prod_{ik}\int\limits_{\HHH_{x_{ik-1}}(s_{ik})} \hspace{-6mm} d^3x_{ik} \Biggr)  L(\ux)^\dagger L(\ux)=I\,.
\ee
While Fubini's theorem allows us to exchange the order of integration, it must be noted here that the domain for $x_{ik}$ depends on $x_{ik-1}$, so the $x_{ik}$-integral must occur to the right of the $x_{ik-1}$-integral. This limitation on the possible ordering of the integrals must be kept in mind; note also that the order of integrals is not a priori related to the order of factors in $L(\ux)$.

Fix the $s_{ik}$ and let $\HHH_{ik}:= \HHH_{x_{ik-1}}(s_{ik})$.
We split the multiple integral, corresponding to the partition of each $\HHH_{ik}$ into the ${}^3 C_{i\uk}$'s with $k_i=k$, into a sum
\be\label{summands}
\sum_{\substack{\uk^{11}\ldots \uk^{Nn_N}\\k^{ik}_i = k \, \forall ik}} \Biggl(\prod_{ik}\int\limits_{{}^3C_{i\uk^{ik}}} \hspace{-3mm} d^3x_{ik} \Biggr) L(\ux)^\dagger L(\ux) \,.
\ee
Each summand is associated with a certain element of ${}^4\sA^\nu$, and different summands with different elements of ${}^4\sA^\nu$. 

As a preparation for the general procedure, let us outline the first step of the induction. In each summand, consider the two innermost $K$ factors of $L(\ux)^\dagger L(\ux)$, let them be $K(x_{j\ell})^\dagger K(x_{j\ell})$. We want to integrate them out using \eqref{Gnormalized}, resulting in a factor $P^{j\ell}_{{}^3C_{j\uk^{j\ell}}}$ in the abbreviated notation
\be\label{P3cell}
P^{j\ell}_{{}^3C_{i\uk}}:= U^0_{\HHH_{ik_i}} \, P_{\HHH_{ik_i}}\bigl(1_{x_{j\ell}\in {}^3 C_{i\uk}}\bigr) \, U^{\HHH_{ik_i}}_0
\ee
(which depends only on the 3-cell rather than on $\HHH_{ik_i}$ by interaction locality).
But let us be slow and integrate out, at this step, $x_{j\ell}$ \emph{only} if ${}^3 C_{j\uk^{j\ell}}$ lies on $\partial {}^4C_{\un}$ (the futuremost surface formed by 3-cells). Since $k^{j\ell}_j=\ell$, and $\HHH_{j\ell}$ must border on ${}^4 C_{\un}$, $\ell=n_j$; thus, there is no $\HHH_{j\ell+1}$, and therefore no obstacle to changing the order of integration so that the rightmost integral is over $x_{j\ell}$. That is, such an $x_{j\ell}$ can, in fact, be integrated out. Since factors corresponding to different 3-cells on $\partial {}^4C_{\un}$ commute, we can integrate them all out. As a result, in each summand, there is no integration any more over any 3-cell on $\partial {}^4C_{\un}$, but for each 3-cell on $\partial {}^4C_{\un}$ involved in a summand, there is a factor of the form \eqref{P3cell}. The induction step will be about considering surfaces made up of 3-cells that lie further and further in the past, until we are done with the pastmost surface $\partial {}^4C_{0^N}$ and all variables are integrated out. Now we give the details. 

Fix an admissible sequence $(V_1,\ldots,V_{r+1})$. We will consider the sequence backwards and count down the index $n$ of $V_n$ from $r+1$ to 1. We write $V_n^c:= {}^4\sA \setminus V_n$ for the complement of $V_n$; since the corresponding space-time set $S(V_n^c)=S(V_n)^c$ is future complete, it is for every $n\neq r+1$ the future of some spacelike surface
\be 
\Sigma_n:= \partial S(V_n^c)= \partial S(V_n) = S(\partial V_n^c) = S(\partial V_n)\,.
\ee
At each stage of the process, each summand stemming from \eqref{summands} is related to $\nu$ 3-cells. After integrating out one variable, we call the associated 3-cell an \emph{out-cell}, while a 3-cell associated with a variable that has not yet been integrated out will be called an \emph{in-cell}. In each step $n\to n-1$ of the induction, we will operate on the summands keeping their sum the same. Let $\sB$ denote the set of abstract flashes, i.e., of all pairs $(i,k)$:
\be\label{Bdef}
\sB:=\bigl\{(i,k)\in \{1\ldots N\} \times \NNN: 1\leq k \leq n_i\bigr\}\,.
\ee
The abstract 3-cells on $\HHH_{ik}$ form the set
\be
{}^3\sA_{ik} = \bigl\{ (i,\uk) \in {}^3\sA: k_i=k \bigr\}\,.
\ee
Moreover, we define the sets that will turn out to be the sets of all in-cells (out-cells, respectively) by
\begin{subequations}\label{inoutdef}
\begin{align}
\In_{nik}&:= \Bigl\{ (i,\uk) \in {}^3\sA_{ik} : \text{earlier than }\partial V_n \Bigr\}\,,\label{indef}\\
\Out_{nik}&:= \Bigl\{ (j,\uell) \in \partial V_n : \text{no later than } {}^3\sA_{ik}\Bigl\} \,.\label{outdef}
\end{align}
\end{subequations}
They correspond to the space-time sets
\be
S(\In_{nik}) = \HHH_{ik} \cap (\past(\Sigma_n)\setminus \Sigma_n)\,,~~~
S(\Out_{nik}) = \Sigma_n \cap \past(\HHH_{ik})\,.
\ee
Together, they form the spacelike surface $\partial \bigl( \future(\HHH_{ik}) \cup \future(\Sigma_n) \bigr)$; see Figure~\ref{fig:inout}.

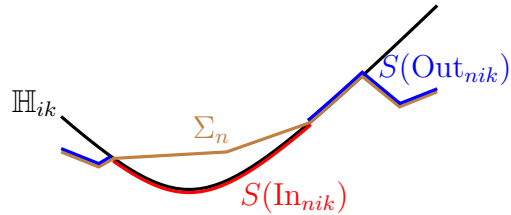
\begin{figure}[h]
\begin{center}
\begin{tikzpicture}
\draw [very thick, domain=-2.5:2.5, samples=50] plot (\x, {-0.5+sqrt(1+(\x+0.8)*(\x+0.8)))} );
\node at (-2.85,1.65) {$\HHH_{ik}$};
\draw[very thick, brown] (-2.5,1) -- (-2,0.8) -- (-1.8,0.914) -- (-0.3,1) -- plot[domain=0.8:1.5, samples=30] (\x, {-0.5+sqrt(1+(\x+0.8)*(\x+0.8))} ) -- (2,1.6) -- (2.5,1.8);
\node[brown] at (-0.5,1.3) {$\Sigma_n$};
\draw[very thick, blue] (-2.5,1.045) -- (-2,0.845) -- (-1.84,0.94);
\draw[very thick, blue] plot[domain=0.78:1.503, samples=30] (\x, {-0.45+sqrt(1+(\x+0.8)*(\x+0.8))} ) -- (2.005,1.645) -- (2.5,1.84);
\node[blue] at (2.6,2.1) {$S(\Out_{nik})$};
\draw[very thick, red, domain=-1.81:0.82, samples=30] plot (\x, {-0.541+sqrt(1+(\x+0.8)*(\x+0.8))} );
\node[red] at (0.6,0.4) {$S(\In_{nik})$};
\end{tikzpicture}
\end{center}
\caption{The sets In and Out for a particular hyperboloid $\HHH_{ik}$ and $\Sigma_n$. Some pieces of hyperbola are drawn as straight lines. Some lines are drawn next to each other (rather than on top of each other) for better visibility.}
\label{fig:inout}
\end{figure}

\medskip

\noindent{\bf Induction hypothesis:} {\it The summands are labeled by the elements of
\be\label{Mndef}
M_n := \Bigl\{ \theta:\sB\to {}^3\sA~:~ \forall ik\in \sB: \theta(ik) \in \In_{nik} \cup  \Out_{nik} \Bigr\}=\prod_{ik} (\In_{nik} \cup  \Out_{nik})
\ee
with $\prod$ the Cartesian product, and the summand labeled $\theta$ reads
\be\label{summandtheta}
\Biggl(\prod_{\substack{ik\in \sB\\ \theta(ik)\in \In_{nik}}} \int\limits_{{}^3 C_{\theta(ik)}} \hspace{-3mm} d^3x_{ik} \Biggr) 
\Biggl( \prod_{\substack{ik\in \sB\\ \theta(ik)\in \In_{nik}}} \hspace{-3mm} K(x_{ik})\Biggr)^{\!\!\dagger}
\Biggl(\prod_{\substack{ik\in \sB\\ \theta(ik)\in \Out_{nik}}} \hspace{-3mm} P^{ik}_{{}^3C_{\theta(ik)}}\Biggr)
\Biggl( \prod_{\substack{ik\in \sB\\ \theta(ik)\in \In_{nik}}} \hspace{-3mm} K(x_{ik}) \Biggr)\,,
\ee
where the product over $K(x_{ik})$ is understood in the order from right to left in which the 3-cells are crossed in the admissible sequence.}

\medskip

The form \eqref{summandtheta} of the summand labeled $\theta$ means, in particular, that the in-cells are integrated over, and the out-cells appear only in the projections in the middle. The order of the factors $P^{ik}_{{}^3C_{\theta(ik)}}$ need not be specified: they commute pairwise because all of the 3-cells ${}^3C_{\theta(ik)}$ lie on a common spacelike surface $\Sigma_n$.

The anchor of the induction is the case $n=r$, in which $\In_{rik}= \{(i,\uk)\in {}^3\sA: k_i=k\}$ and $\Out_{rik}=\emptyset$, so $M_r$ as in \eqref{Mndef} corresponds to those $(\uk^{11}\ldots \uk^{Nn_N})$ with $k^{ik}_i=k$, and the summands agree with those of \eqref{summands}.

On the other end, for $n=2$, we find that $V_2=\{0^N\}$, $\Sigma_2=\partial {}^4C_{0^N}$, $\In_{2ik}=\emptyset$, and $\Out_{2ik}$ contains exactly the 3-cells on $\Sigma_2$. So the induction hypothesis, when proved, will imply that no summands involve integrals any more, and the sum reads
\be
\sum_{\uk^{11}\ldots \uk^{Nn_N} \in \partial \{0^N\}}
\prod_{ik\in\sB}  P^{ik}_{{}^3C_{\uk^{ik}}}
= \prod_{ik\in \sB} \underbrace{\Biggl( \sum_{\uk \in \partial \{0^N\}} P^{ik}_{{}^3C_{\uk}}\Biggr)}_{=I} =I\,,
\ee
as needed for \eqref{toshow2}.

So it remains to carry out the {\bf induction step} $n\to n-1$, which consists of two parts. The first part deals with the projections in the middle of \eqref{summandtheta}, the second with integrating out some of the variables.

{\it First part:} Interaction locality in the form \eqref{ILB} implies that the projection to the future boundary of a 4-cell can be ``pulled across'' the 4-cell, i.e., is equal to the projection to its past boundary, 
\be\label{pullacross}
P^{j\ell}_{\partial_+ {}^4C}= P^{j\ell}_{\partial_- {}^4C}\,,
\ee
where $\partial_\pm$ denotes the future (past) boundary, which consists of one or more 3-cells, and $P^{j\ell}_{\partial_{\pm} {}^4C}$ equals the sum of the $P^{j\ell}_{{}^3C}$ over all 3-cells ${}^3C$ belonging to $\partial_{\pm} {}^4C$. The relevant 4-cell ${}^4C$ here is the one crossed by the admissible sequence between $n-1$ and $n$, ${}^4C={}^4C_{\uk}$ with $V_n = V_{n-1} \cup \{\uk\}$. The future boundary of ${}^4C$ consists of 3-cells belonging to $\Sigma_n$, the past boundary of 3-cells belonging to $\Sigma_{n-1}$; in fact, the only difference between $\Sigma_n$ and $\Sigma_{n-1}$ is that the 3-cells belonging to the future boundary of ${}^4C$ are replaced by those belonging to the past boundary of ${}^4C$. 

For every $ik\in \sB$, 
\be\label{allornone}
\text{either all or none of the 3-cells in $\partial_+{}^4C$ belong to $\Out_{nik}$.} 
\ee
Indeed, either ${}^4C\subseteq \future(\HHH_{ik})$ or ${}^4C\subseteq \past(\HHH_{ik})$. In the former case, $\partial_+{}^4C \subseteq \future(\HHH_{ik})\setminus \HHH_{ik}$; since $\Out_{nik}$ lies in the past of $\HHH_{ik}$, none of the 3-cells in $\partial_+{}^4C$ belong to $\Out_{nik}$. In the latter case, $\partial_+{}^4C\subseteq \past(\HHH_{ik})$; since all of the 3-cells in $\partial_+ {}^4C$ belong to $\Sigma_n$, they all belong to $\Out_{nik}$, which proves \eqref{allornone}.

Now define
\be
\widetilde{\Out}_{nik}:=
\begin{cases}
(\Out_{nik}\setminus \partial_+{}^4C)\cup \partial_-{}^4C &\text{if } \partial_+{}^4C \subseteq \Out_{nik}\\
\Out_{nik}&\text{otherwise}
\end{cases}
\ee
and $\widetilde M_n$ like $M_n$ in \eqref{Mndef} but with $\Out_{nik}$ replaced by $\widetilde\Out_{nik}$.
\be\label{claim}
\mbox{\begin{minipage}{0.85\textwidth}
{\it Claim}: The sum over $\theta\in M_n$ of \eqref{summandtheta} equals the sum over $\theta\in \widetilde M_n$ of \eqref{summandtheta} with $\Out_{nik}$ replaced by $\widetilde\Out_{nik}$.
\end{minipage}}
\ee
To see this, think of 
$M_n$ as the rightmost expression of \eqref{Mndef}. We take the following step successively for each $j\ell\in \sB$ (in any ordering of $\sB$): We replace $\Out_{nj\ell}$ in \eqref{Mndef} and \eqref{summandtheta} by $\widetilde\Out_{nj\ell}$; that is, a $P$ factor appears in each summand for each $ik$ for which $\theta(ik)\in\Out_{nik}$ respectively $\theta(ik)\in\widetilde\Out_{nik}$, depending on whether the replacement step has already been done for $ik$. We check that each step leaves the sum unchanged; in fact, for every fixed choice of $\theta(ik)$ for all $ik\neq j\ell$, the sum of the summand remains unchanged. Indeed, this sum is a sum over all $\theta(j\ell)\in \In_{nj\ell}\cup \Out_{nj\ell}$. The summands with $\theta(j\ell)\in \In_{nj\ell}\cup \Out_{nj\ell}\setminus \partial_+{}^4C$ do not change. By \eqref{allornone}, the summands with $\theta(j\ell)\in \Out_{nj\ell}\cap\partial_+{}^4C$ together are either 0 or can be combined into one expression of the form \eqref{summandtheta} with $P^{j\ell}_{{}^3C_{\theta(j\ell)}}$ replaced by $P^{j\ell}_{\partial_+{}^4C}$. By \eqref{pullacross}, $\partial_+$ can be replaced by $\partial_-$, and by the same reasoning backwards, this equals the sum over $\theta(j\ell)\in \widetilde\Out_{nj\ell}\cap\partial_-{}^4C$. Thus, each step leaves the sum unchanged, and after all steps (for all $j\ell$), we have proved the  claim \eqref{claim}.

At this point, we have achieved in particular that all $P$ factors refer to 3-cells on $\Sigma_{n-1}$.

{\it Second part:} We now wish to integrate out all variables that vary over 3-cells in $\Sigma_{n-1}$. We can do this for each summand individually, so focus on a particular $\tilde\theta\in \widetilde M_n$. The only 3-cells in $\Sigma_{n-1}$ that were not included already in $\Sigma_n$ are those in $\partial_-{}^4C$, and the only ones that any variable $x_{ik}$ ever gets integrated over are those in $\In_{nik}$. In the given summand $\tilde\theta$, there can be none or one or several variables $x_{ik}$ for which $\partial_-{}^4C$ overlaps with $\In_{nik}$. If none, we leave the summand unchanged. If one or more, we will treat them successively in an arbitrary order. So let $x_{j\ell}$ be one of them. The leftmost factors in the $\prod K(x_{ik})$ in \eqref{summandtheta} are those referring to 3-cells in $\partial_-{}^4C$; by Proposition~\ref{prop:3Ccommute}, these factors commute with each other, so we can assume that $K(x_{j\ell})$ is the leftmost one.

Now we want to make sure that the integral over $x_{j\ell}$ is the rightmost integral. We can change the order of integration using Fubini's theorem, provided the domains of integration of the other integration variables do not depend on $x_{j\ell}$. The variables whose domain depends on $x_{j\ell}$ are $x_{j\ell+1}$ and higher ones for particle $j$. Since these domains all lie on $\HHH_{j\ell+1}$ or later, and thus in the future of $x_{j\ell}$, they lie on $\Sigma_n$ or later, so by \eqref{outdef} and the induction hypothesis, all of these variables have already been integrated out in previous induction steps ($\In_{nj\ell+1}=\emptyset$), and we can assume that the $x_{j\ell}$ integral is the rightmost integral.

For carrying out the integral, we need that the space-time locations of the 3-cells ${}^3C_{\tilde\theta(ik)}$ in the factors $P^{ik}_{{}^3C_{\tilde\theta(ik)}}$ do not depend on $x_{j\ell}$. This follows if none of these 3-cells lies in the strict (open) future of $x_{j\ell}$. Now all of these 3-cells lie on $\Sigma_{n-1}$, a spacelike hypersurface containing $x_{j\ell}$, and thus not in the strict future of $x_{j\ell}$. 

We also need that $K(x_{j\ell})$ commutes with the $P$'s in the middle. That is the case because the $P$'s are multiplication operators on their 3-cells and thus (by interaction locality) on $\Sigma_{n-1}$; likewise, $K(x_{j\ell})$ is by its definition \eqref{Kdef} a multiplication operator on the 3-cell ${}^3C(x_{j\ell})$ containing $x_{j\ell}$ (which remains the same 3-cell ${}^3C:= {}^3C_{\tilde\theta(j\ell)}$ during the integration over $x_{j\ell}$) and thus (by interaction locality) a multiplication operator on $\Sigma_{n-1}$. Since all multiplication operators on a common spacelike surface commute, we can pull $K(x_{j\ell})$ to the left of all $P$'s, where it arrives next to $K(x_{j\ell})^\dagger$. Since none of the other factors ($P$'s and $K$'s) in the integrand depends on $x_{j\ell}$, they can be pulled out of the $x_{j\ell}$ integral. By \eqref{Gnormalized}, the integral can be carried out to yield
\be
\int_{{}^3C} \!\!\! d^3 x_{j\ell} \, K(x_{jk})^\dagger K(x_{jk}) = P^{j\ell}_{{}^3C}\,.
\ee
This factor joins the $P$ factors, showing up in the correct position among all factors in the remaining integrand \eqref{summandtheta}. In particular, still all $P$ factors refer to 3-cells \emph{on $\Sigma_{n-1}$}. We repeat this operation of carrying out the integral for all integrals over 3-cells on $\partial_-{}^4C$. Afterwards, in this summand $\tilde\theta$ the out-cells (with $P$ factors) are those in $\widetilde\Out_{nik}$ together with those in $\partial_-{}^4C$, and thus exactly those in $\Out_{n-1,ik}$; the in-cells (with $K$ factors) are those in $\In_{nik}$ except for those on $\Sigma_{n-1}$ or later, and thus exactly those in $\In_{n-1,ik}$. The summand has the form \eqref{summandtheta} with $n$ replaced by $n-1$, and the index $\theta$ labeling the summands runs through $M_{n-1}$. We have thus proved the induction hypothesis for $n-1$, completed the induction step, and completed the proof of Proposition~\ref{prop:normalized}.
\end{proof}

\subsubsection{Definition of the Theory}

We have defined the model in \eqref{PPPdef} for chosen numbers $n_i$ of flashes for each particle $i$. If we want to think of this model as a theory of the universe, and compare it to our empirical observations, we should take the limit $n_i\to \infty$ or choose $n_i$ very large. 

In contrast to the non-interacting 2004 model, in the present model the marginal distribution of the first $\tilde n_i$ flashes for each particle $i$ (i.e., the distribution after integrating out the flashes after $\tilde n_i$) is \emph{not} given by the same formula \eqref{PPPdef}, although it is still given by \emph{some} POVM. That is because the partition of the hyperboloids into 3-cells depends on the later flashes, and thus so does the procedure of cutting off the tails of the Gaussians. As a consequence, for the 2004 model we did not actually have to specify the numbers $n_i$, but now we have to; any choice of very large $n_i$ should yield reasonable behavior of the theory, as well as the limit $n_i\to\infty$.

\section{Properties of the Model}
\label{sec:properties}

\begin{enumerate}
\item {\it Size of 3-cells.} Since the tails of the Gaussian profile function get cut off at the boundary of a 3-cell $A$, the width of the resulting profile function $g_{yAx}$ could be smaller than $\sigma$ if the diameter of $A$ is, which could have undesirable consequences such as amplified empirical deviations of the model from standard quantum mechanics. I have made a crude estimate of the typical diameter of the 3-cells for condensed matter under everyday conditions and arrived at several millimeters or larger, which is much larger than GRW's suggested value of $\sigma = 10^{-7}$ m and thus suggests that the deviations are not amplified. Put differently, the tails are typically cut off at about $10^4$ standard deviations, so the change is tiny. A more careful study of this question would be of interest.

Matthias Lienert has made the interesting suggestion (personal communication) that since the Gaussians get cut off anyway, maybe they can be dispensed with altogether and replaced by a constant function (corresponding to the limit $\sigma \to\infty$); at each collapse, the wave function would then be localized to the size of a 3-cell. An investigation of whether such a theory is viable would be of interest.

\item {\it Stochastic evolution of the wave function.} In order to define a theory with flash ontology, it suffices to define the joint distribution of the flashes. But it is common to think of collapse model in terms of a stochastically evolving wave function. Such a wave function $\psi_\Sigma$ can be defined for the present model for every spacelike surface $\Sigma$ as follows. It should be related to
the conditional probability distribution of $\uX$, given the flashes up to  $\Sigma$. To express this distribution, let $\sI\subseteq \sB$ be an arbitrary index set of $ik$'s (with $\sB$ as in \eqref{Bdef} the set of all $ik$'s), let $\sI^c:= \sB\setminus \sI$, and let $\uX_{\sI}$ be the collection of $X_{ik}$ with $ik\in\sI$; likewise $\ux_{\sI}$ etc., so we can write $\ux=(\ux_{\sI},\ux_{\sI^c})$. Then the conditional distribution of the flashes after $\Sigma$, given that those before $\Sigma$ were at $\ux_{\sI}\in\past(\Sigma)^{\sI}$, is
\begin{multline}\label{condprob}
\PPP\biggl(\uX_{\sI^c}\in d\ux_{\sI^c} \bigg| \uX_{\sI^c}\in \future(\Sigma)^{\sI^c} \text{ and } \uX_{\sI}=\ux_{\sI} \biggr) 
=\\ 
1_{\ux_{\sI^c}\in\future(\Sigma)^{\sI^c}} \: \frac{\scp{\psi_0}{D(\ux)|\psi_0}}{\scp{\psi_0}{W_\Sigma(\ux_{\sI})^2|\psi_0}} \: d\ux_{\sI^c}
\end{multline}
with positive operators
\be\label{Wdef}
W_\Sigma(\ux_{\sI})= \biggl( \frac{G(d\ux_{\sI}\times \future(\Sigma)^{\sI^c})}{d\ux_{\sI}}\biggr)^{1/2} = \Biggl(\:\int\limits_{\future(\Sigma)^{\sI^c}} \hspace{-6mm} d\ux_{\sI^c} \,  D(\ux)\Biggr)^{1/2}\,.
\ee
(The condition that $X_{ik+1}\in \future(X_{ik})$ restricts the relevant index sets $\sI$, but this fact does not change the validity of \eqref{condprob}.) We therefore define, given that the flashes up to $\Sigma$ were $\ux_{\sI}$,
\be\label{psiSigmadef}
\psi_\Sigma:=\frac{U^\Sigma_0 \, W_\Sigma(\ux_{\sI}) \psi_0}{\| W_\Sigma(\ux_{\sI}) \psi_0\|}\,,
\ee
in analogy to Eq.~(25) of \cite{Tum06c}. Considering a fixed pattern $\ux$ of flashes and varying $\Sigma$, this wave function changes abruptly whenever $\Sigma$ crosses one of the flashes (as $\sI$ changes then). The conditional probability \eqref{condprob} can be expressed as
\be\label{condprob2}
\PPP= 
1_{\ux_{\sI^c}\in\future(\Sigma)^{\sI^c}} \: \Bscp{\psi_\Sigma}{U^\Sigma_0 W_\Sigma(\ux_{\sI})^{-1}D(\ux)W_\Sigma(\ux_{\sI})^{-1}U^0_\Sigma\Big|\psi_\Sigma} \: d\ux_{\sI^c}\,.
\ee

\item {\it Non-interacting special case.} If the given unitary hypersurface evolution $U^{\Sigma'}_\Sigma$ is non-interacting, the situation simplifies as different particle variables $x_j$ in the wave function can be evolved to different surfaces, and $K(x_{ik})$ commutes with $K(x_{j\ell})$ for $j\neq i$. If we could replace the cut-off Gaussians $g_{yAx}$ of \eqref{gAdef} in the definition \eqref{Kdef} of the collapse by the original Gaussians $\tilde g_{yx}$ of \eqref{tildegdef}, we would obtain exactly the 2004 model. Thus, whenever it is the case that the 3-cells $A$ are typically much larger than the width $\sigma$ of the Gaussians, then (with high probability) the cutting off does not make a big difference as it concerns only tiny tails of the Gaussian, and the 2004 model is a close approximation to the non-interacting case of the present model.

\item {\it Non-locality.} The collapse model presented here is non-local while being fully relativistic. In fact, it violates Bell's inequality. The non-locality corresponds to the fact that the joint distribution of two flashes is not a product even when the flashes are spacelike separated. Already the 2004 model was non-local, and further aspects of this property were discussed in \cite{Tum06,Tum06c,Tum07,Tum09}.

\item {\it Microscopic parameter independence.} This is the property of a theory that the probability distribution of the local beables before any spacelike surface $\Sigma$ does not depend on the external fields after $\Sigma$. For example, microscopic parameter independence is grossly violated in Bohmian mechanics (for $\Sigma$ not belonging to the preferred foliation). The model presented here does not satisfy microscopic parameter independence exactly, but it does up to small deviations. 

This is suggested by the following considerations. First, $U_0^{\Sigma'}$ does not depend on the external fields after $\Sigma$ if both $\Sigma_0$ and $\Sigma'$ lie in the past of $\Sigma$; by interaction locality, a collapse operator $K(x_{ik})$ does not depend on the external fields after $\Sigma$ if ${}^3C(x_{ik})$ lies in the past of $\Sigma$. The space-time location of ${}^3C(x_{ik})$ (specifically, where its boundaries are) depends on other $x_{j\ell}$, but only on those before $\Sigma$. As a by-product of the proof of Proposition~\ref{prop:normalized}, the marginal distribution of the flashes in the past of $\partial V_n$ for $V_n\in\sN$ is given by the sum over $\theta\in M_n$ of the integrands in \eqref{summandtheta}, so if $S(\partial V_n)$ lies in the past of $\Sigma$, this distribution will not depend on external fields after $\Sigma$. However, even if $x_{ik}$ lies in the past of $\Sigma$, ${}^3C(x_{ik})$ need not lie in the past of $\Sigma$. Yet, it seems that the significant support of $g_{x_{ik-1},{}^3C(x_{ik}),x_{ik}}$ reaches no further than about $\sigma/c\approx 10^{-15}$ s into the future of $\Sigma$.

\item {\it No signaling.} This property means the impossibility for agents  to transmit messages faster than light; it should follow from microscopic parameter independence, as the message to be sent could be modeled as an external field and the message received would have to be some (coarse-grained) function of the local beables.

\item {\it Non-relativistic limit.} In the non-relativistic limit, the present model reduces to the non-relativistic GRW model, provided that the unitary evolution reduces to a non-relativistic unitary evolution. To see this, note that in the limit the hyperboloids become horizontal 3-planes, while the intersection between two hyperboloids escapes to infinity, so that every 3-cell becomes a full horizontal 3-plane and every 4-cell a layer between two such planes. Thus, cutting off the Gaussians becomes irrelevant, there is only one admissible sequence, $K$ is just the Heisenberg-evolved multiplication by a Gaussian, and it becomes visible that the joint distribution of the flashes approaches that of the non-relativistic GRW model.
\end{enumerate}

\end{document}